\documentclass[a4paper]{article}

\usepackage[english]{babel}
\usepackage{fullpage}
\usepackage{amsmath}
\usepackage[colorlinks=true, allcolors=blue]{hyperref}
\usepackage{amssymb,amsthm}
\usepackage{bbm}
\usepackage[capitalize]{cleveref}
\usepackage{tikz}
\usepackage{authblk}
\usepackage{mathtools}
\usepackage{orcidlink}

\usepackage[sorting=none, backend=biber, url=false, isbn=false, hyperref=true, eprint=true, maxbibnames=6]{biblatex}
\AtEveryBibitem{%
  \clearfield{eprintclass}%
}

\newbibmacro{string+doi}[1]{%
  \iffieldundef{doi}{#1}{\href{http://dx.doi.org/\thefield{doi}}{#1}}}
\DeclareFieldFormat{title}{\usebibmacro{string+doi}{\mkbibemph{#1}}}
\DeclareFieldFormat[article]{title}{\usebibmacro{string+doi}{\mkbibquote{#1}}}

\newcommand{\ket}[1]{\vert #1 \rangle}
\newcommand{\bra}[1]{\langle #1 \vert}
\newcommand{\scalprod}[2]{\langle #1 \vert #2 \rangle}
\newcommand{\tr}{\operatorname{Tr}}
\newcommand{\id}{\mathbbm{1}}
\newcommand{\Span}{\mathrm{Span}}
\newcommand{\SO}{\mathrm{SO}}

\newcommand{\ad}{\mathrm{ad}}

\newcommand{\su}{\mathfrak{su}}
\newcommand{\so}{\mathfrak{so}}
\newcommand{\End}{\mathrm{End}}

\newtheorem{theorem}{Theorem}

\newtheorem{proposition}[theorem]{Proposition}

\newtheorem{definition}[theorem]{Definition}
\theoremstyle{definition}
\newtheorem{remark}[theorem]{Remark}

\tikzset{
  tensor/.style={
    inner sep = 0.055cm,
    shape = circle,
    draw,
    fill
  }
}

\title{Spin models with critical ground space degeneracy from Lie algebra relations}
\author[1]{András Molnár\,\orcidlink{0000-0001-5144-438X}\,}
\author[3]{Efekan Kökcü\,\orcidlink{0000-0002-7323-7274}\,}
\author[1,2]{Norbert Schuch\,\orcidlink{0000-0001-6494-8616}\,}
\author[4]{Bojko N. Bakalov\,\orcidlink{0000-0003-4630-6120}\,}
\affil[1]{\mbox{\normalsize University of Vienna, Faculty of Mathematics, Oskar-Morgenstern-Platz 1, 1090 Vienna, Austria}}
\affil[2]{\mbox{\normalsize University of Vienna, Faculty of Physics, Boltzmanngasse 5, 1090 Vienna, Austria}}
\affil[3]{\mbox{\normalsize Department of Electrical and Computer Engineering, University of Central Florida, Orlando, FL 32816, USA}}
\affil[4]{\mbox{\normalsize Department of Mathematics, North Carolina State University, Raleigh, NC 27695, USA}}

\begin{document}
\maketitle

\begin{abstract}
We introduce an $\mathrm{SO}(3)$-symmetric spin model derived from
the algebraic structure of $\mathfrak{so}(3)$: It is obtained as the
nearest-neighbor parent Hamiltonian of a Matrix Product State (MPS) built from
the generators of the Lie algebra $\mathfrak{so}(3)$ and the identity matrix, and
therefore encodes the quadratic relations of the Lie algebra. We characterize
the ground space structure of the resulting model and show that for open
boundary conditions (OBC), it exhibits a quadratic ground space degeneracy, given by one 
irreducible representation (irrep) of each odd dimension (integer spin). For periodic
boundary conditions (PBC), it has a linear ground space degeneracy,
consisting of a singlet --- namely, the MPS underlying the model --- and a
ferromagnet (the irrep with maximal spin), and thus exhibits gapless
excitations.  We also generalize the construction and the analysis of the OBC ground space 
to $\SO(d)$ and other compact simple Lie groups.

The MPS on which the model is based is an injective MPS, and therefore, its
3-site parent Hamiltonian has a unique ground space and is gapped. The observed
critical behavior thus has its origin in the small parent Hamiltonian considered. The
model's algebraic ground space scaling thus provides an example distinct from
that reported in (Schuch \emph{et al.}, 2025 \cite{schuch2025}), where it
was shown that small parent Hamiltonians of injective MPS generically have unique ground
states, and examples with an exponential ground space degeneracy were given.
\end{abstract}

\section{Introduction}\label{sec:introduction}

Matrix Product States  (MPS) provide a way to concisely describe wave functions in extended quantum many-body systems \cite{Cirac2021}. They are especially well suited for approximating ground and Gibbs states of local, gapped Hamiltonians \cite{Hastings2006a,Molnar_2015}. They have also been used in several other contexts, such as machine learning \cite{Sengupta2022}, simulation of quantum algorithms and time evolution \cite{Paeckel2019}, quantum chemistry \cite{Garcia2024}, etc.

MPS arise as special cases of Finitely Correlated States \cite{Fannes1992}, a family of well-behaved states defined directly in the thermodynamic limit, and one can readily reformulate the Density Matrix Renormalization Group algorithm \cite{White1992} using MPS. These two facts are the underlying reason for their usefulness: both as the basis for highly efficient numerical algorithms and as a tool to create and investigate Hamiltonians with known ground-space and spectral properties; in particular, exactly solvable models. 

The main idea behind the use of MPS as a tool to create exactly solvable models is the generalization of the AKLT model \cite{Affleck1988}, where a Hamiltonian is devised whose ground state is a known, already fixed state. A similar construction works in the MPS language: Given an MPS, one can create a Hamiltonian, called the parent Hamiltonian \cite{Fannes1992}, such that the MPS is in its ground space. Under certain conditions\footnote{For example, the MPS tensor is injective and the range of the Hamiltonian terms is larger than the injectivity length.} one can prove that this Hamiltonian has a unique ground state and a gap separating its ground state energy from the rest of the spectrum \cite{Fannes1992,Nachtergaele1996} (with a non-zero lower bound independent of the system size, i.e., the infinite Hamiltonian also has a gap). 

Symmetries of an MPS (and, consequently, of its parent Hamiltonian) can be guaranteed by enforcing the symmetries of its local building blocks \cite{Fannes1992,Perez-Garcia2007}. This fact allows one to characterize, for example, symmetry-protected topological order in one-dimensional systems \cite{1010.3732,Chen2011a}. Understanding how MPS describe symmetric states also speeds up numerical algorithms calculating ground state properties of Hamiltonians possessing symmetries, which is a common trait of realistic Hamiltonians.  

In this paper, we devise an MPS and a corresponding parent Hamiltonian which are
derived from the Lie algebra structure of the symmetry group $\mathrm{SO}(3)$.
Specifically, the MPS is constructed from the generators of the underlying Lie algebra, 
together with the identity matrix, where the Hamiltonian we consider is its nearest
neighbor parent Hamiltonian; as such, the Hamiltonian, and thus its ground
space, precisely encodes the quadratic relations of the Lie algebra.  The model
possesses a number of interesting properties: First, it exhibits an
$\mathrm{SO}(3)$ symmetry under the adjoint action, i.e., the spin $0\oplus 1$
representation. By construction, the ground space contains the MPS from which
the parent Hamiltonian has been derived. However, as we prove, the ground space
of the model is significantly larger: For OBC, it displays a quadratically
growing degeneracy, which decomposes into of one copy of each integer spin
representation of $\mathrm{SO}(3)$. For PBC, its ground space degeneracy grows
linearly, and the ground space consists of the defining MPS, together with the
maximum spin representation. The latter corresponds to a ferromagnetic ground
space, and therefore, the model exhibits gapless spin wave excitations. 

Importantly, the MPS from which the model is derived exhibits a property
known as injectivity after blocking two sites, which implies that its parent
Hamiltonian constructed on three consecutive sites (i.e., one above the
injectivity length) possesses a unique ground state for PBC (and a constant
degeneracy for OBC)~\cite{Perez-Garcia2007}.  The observed critical behavior
of our model
thus arises from the fact that the Hamiltonian we consider is a small parent
Hamiltonian, that is, it is constructed on two sites (i.e., the injectivity length) only. Additionally, it was
recently shown that for  generically chosen injective MPS, such a small parent
Hamiltonian still has a unique ground state, and deviating cases were devised only
for exponentially growing ground space degeneracy~\cite{schuch2025}. The model
introduced in this work thus provides an example of yet another type of behavior, where
the ground space degeneracy of the model exhibits a critical (that is,
algebraic) scaling. While  such a quadratic ground-space dimension scaling can
also be observed in the Motzkin chain in the absence of boundary
terms~\cite{Movassagh2014}, our model differs in that it possesses an underlying
symmetry group, which is directly tied to the observed degeneracy, and in that
the model is obtained as the parent Hamiltonian of an MPS.

We generalize this construction to simple compact Lie groups, resulting in a series of models with various ground-space degeneracies. The proof of the ground-space degeneracy formula builds on Vogel's universal parameters characterizing the eigenvalues of the adjoint action of the split Casimir element on the symmetric square of a simple Lie algebra \cite{vogel1999} (see also \cite{isaev2026vogel}).

Finally, let us note that the open boundary condition ground spaces of frustration-free Hamiltonians, such as the parent Hamiltonians of MPS, naturally form a subproduct system \cite{ShalitSolel2009}. The subproduct systems we obtain are Lie-algebra equivariant by construction, and here we analyze how this equivariance governs the growth of the dimensions.

The paper is organized as follows. First, in \Cref{sec:background}, we introduce the background necessary to understand the tensor network formalism and state the exact theorems known about parent Hamiltonians of MPS. In \Cref{sec:so_3}, we introduce the $\SO(3)$-symmetric MPS that we investigate and analyze its ground-space structure on open boundary conditions. In \Cref{sec:general}, we show how to generalize the construction to arbitrary Lie algebras. In \Cref{sec:pbc}, we analyze the ground-space degeneracy of the $\SO(3)$ model for periodic boundary conditions. \Cref{sec:physical} analyzes the physical structure of the model. Finally, \Cref{sec:conclusion} provides concluding remarks and outlines directions for future work.

\section{Background on MPS}\label{sec:background}

In this section, we briefly introduce MPS and their parent Hamiltonians. For details, see, for example, \cite{Cirac2021}. In the following, $\mathcal{M}_D$ denotes the algebra of $D\times D$ complex matrices. 

\begin{definition}[Matrix Product States]
    Let $\mathcal{H}$ be a finite-dimensional Hilbert space, and let us fix a basis $\{\ket{i}\}_{i=1}^d$ for $\mathcal{H}$. A tensor $A\in\mathcal{M}_D \otimes \mathcal{H}$ of the form 
    \begin{equation*}
    A = \sum_{i=1}^d A_i \otimes \ket{i}
    \end{equation*}
    is called an MPS tensor. The \emph{MPS generated by $A$ with boundary $X$} on $n\in\mathbb N$ particles is a state 
    $\mathfrak{M}_n(A,X)\in \mathcal{H}^{\otimes n}$ defined as 
    \begin{equation*}
        \mathfrak{M}_n(A,X) = \sum_{i\in \{1, \dots, d\}^n} \tr\left( X A_{i_1} A_{i_2} \cdots A_{i_n}\right) \ket{i_1 i_2 \cdots i_n}.
    \end{equation*}
    We write $\mathfrak{M}_n(A)$ for $\mathfrak{M}_n(A,\id)$ and call it the \emph{MPS generated by $A$} on $n$ particles. The \emph{subspace generated by $A$} on $n$ particles is the space
    \begin{equation*}
        \mathfrak{S}_n(A) = \left\{ \mathfrak{M}_n(A,X) \,\middle| \, X\in\mathcal{M}_D\right\} \subseteq \mathcal{H}^{\otimes n}.
    \end{equation*}
\end{definition}

One can use a graphical language to represent tensors and tensor contraction. In this language, tensors are represented by various shapes with lines joined to them; the number of lines corresponds to the number of indices in the tensor. For example, the MPS tensor $A\in\mathcal{M}_D \otimes \mathcal{H}$, $A = \sum_i A_i \otimes \ket{i}$ is a three-index tensor and is represented as 
\begin{equation}
    A = 
    \begin{tikzpicture}[baseline=-1mm, scale = 0.5]
        \draw (-1, 0) -- (1,0);
        \draw (0,0) --++ (0,1);
        \node[tensor] at (0,0) {};    
    \end{tikzpicture}\ ,
\end{equation}
where the two horizontal lines correspond to the indices of the matrices $A_i$ (left is the index corresponding to the output, right is to the input of the matrix), while the vertical line corresponds to the physical space (the vector $\ket{i}$).
Index contraction is denoted by joining lines; with this notation, the contraction of two MPS tensors is depicted as
\begin{equation}
    \sum_{i,j=1}^d A_i A_j \otimes \ket{ij} = \,
    \begin{tikzpicture}[baseline=-1mm, scale = 0.5]
        \draw (-1, 0) -- (2,0);
        \draw (0,0) --++ (0,1);
        \draw (1,0) --++ (0,1);
        \node[tensor] at (0,0) {};    
        \node[tensor] at (1,0) {};    
    \end{tikzpicture}\ ,
\end{equation}
and similarly, the contraction of $n$ MPS tensors multiplied by a matrix $X$ is represented as 
\begin{equation}
    \sum_{i\in \{1, \dots, d\}^n} X A_{i_1} A_{i_2} \cdots A_{i_n} \otimes \ket{i_1 i_2 \cdots i_n} = \,
    \begin{tikzpicture}[baseline=-1mm, xscale = 0.7, yscale = 0.5, font=\footnotesize]
        \draw (-2, 0) -- (4,0);
        \node[tensor, label=below:$X$] at (-1,0) {};            
        \foreach \x in {0,1, 3}{
            \draw (\x,0) --++ (0,1);
            \node[tensor] at (\x,0) {};            
        }
        \node[fill=white] at (2,0) {$\dots$};
    \end{tikzpicture}\ .
\end{equation}
Taking the trace of the first tensor component is just another index contraction, and thus the MPS itself can be represented as
\begin{equation}
    \mathfrak{M}_n(A,X) = \sum_{i\in \{1, \dots, d\}^n} \tr(X A_{i_1} A_{i_2} \cdots A_{i_n})  \ket{i_1 i_2 \cdots i_n} = \;
    \begin{tikzpicture}[baseline=-1mm, xscale = 0.7, yscale = 0.5, font=\footnotesize]
        \draw (-2, 0) rectangle (4,-1);
        \node[tensor, label=below:$X$] at (-1,0) {};            
        \foreach \x in {0,1, 3}{
            \draw (\x,0) --++ (0,1);
            \node[tensor] at (\x,0) {};            
        }
        \node[fill=white] at (2,0) {$\dots$};
    \end{tikzpicture}\ .
\end{equation}

Given an MPS, one can define a Hamiltonian such that the MPS is in the ground space of the Hamiltonian. We will see later that, under certain conditions, the MPS is actually the unique ground state of this Hamiltonian.

\begin{definition}[Parent Hamiltonian]
    Let $A\in\mathcal{M}_D\otimes\mathcal{H}$ be an MPS tensor such that $\mathfrak{S}_k(A)\subsetneq \mathcal{H}^{\otimes k}$. Let $h$ be the orthogonal projector onto $\mathfrak{S}_k(A)^\perp$. The \emph{$k$-body open boundary condition parent Hamiltonian of $A$} is an operator $H \in \mathcal{B}(\mathcal{H}^{\otimes n})$ defined as
    \begin{equation}\label{parentH}
        H = \sum_{i=0}^{n-k} \id^{\otimes i} \otimes h \otimes \id^{\otimes (n-k-i)}.
    \end{equation}
\end{definition}
The $k$-body Hamiltonian term $h$ by definition satisfies $h\geq 0$ and $h\cdot  \mathfrak{M}_k(A,X) = 0$ for all $X\in\mathcal{M}_D$, or equivalently,
\begin{equation*}
    \sum_{i\in \{1, \dots, d\}^k} A_{i_1} A_{i_2} \cdots A_{i_k} \otimes h\ket{i_1i_2 \cdots i_k} = 0.
\end{equation*}
Graphically, this equation is represented as 
\begin{equation}
    \begin{tikzpicture}[baseline=-1mm, xscale = 0.7, yscale = 0.5, font=\footnotesize]
        \draw (-1, 0) -- (4,0);
        \foreach \x in {0,1, 3}{
            \draw (\x,0) --++ (0,1.8);
            \node[tensor] at (\x,0) {};            
        }
        \draw[fill=white] (-0.5,0.7) rectangle (3.5,1.3);
        \node at (1.5, 1.0) {$h$};
        \node[fill=white] at (2,0) {$\dots$};
    \end{tikzpicture}\  = 0.
\end{equation}

The parent Hamiltonian $H$ is a sum of orthogonal projectors, and thus $H\geq 0$. Let $G_n$ be the zero-energy eigenspace of $H$. By construction, $\mathfrak{S}_n(A)\subseteq G_n$, so, as long as $\mathfrak{S}_n(A)\neq 0$, $G_n$ is non-zero and thus the lowest eigenvalue of $H$ is $0$, i.e., $H$ is \emph{frustration free}. In this case, $G_n$ is the \emph{ground space} of $H$. In fact, the ground space can be constructed as follows:
\begin{proposition}
    Assume $A$ is an MPS tensor such that $\mathfrak{S}_n(A)\neq 0$. Then the ground space of the k-body parent Hamiltonian of $A$ on $n$ particles is
    \begin{equation*}
        G_n = \bigcap_{i=0}^{n-k} \mathcal{H}^{\otimes i} \otimes \mathfrak{S}_k(A) \otimes \mathcal{H}^{\otimes (n-k-i)}.
    \end{equation*}
\end{proposition}

\begin{proof}
    We claim that $G_n$ is the zero energy eigenspace of $H$. Indeed, $\bra{v}H\ket{v} = 0$ iff $\bra{v}\id^{\otimes i} \otimes h \otimes \id^{\otimes (n-k-i)} \ket{v}=0$ for all $i=0,\dots, n-k$, which is equivalent to $\ket{v}\in G_n$. As $H\geq 0$, each of its eigenvalues is non-negative, and thus $G_n$ is the ground space of $H$ as long as it is nonzero, $G_n\neq 0$. But by construction, $\mathfrak{S}_n \subseteq \bigcap_{i=0}^{n-k} \mathcal{H}^{\otimes i} \otimes \mathfrak{S}_k \otimes \mathcal{H}^{\otimes (n-k-i)} = G_n$, and hence $G_n \neq 0$. 
\end{proof}

As mentioned in \Cref{sec:introduction}, the parent Hamiltonians of certain MPS have unique ground spaces. These are injective\footnote{In some of the literature, these MPS are called normal, and ``injective'' refers to normal MPS with injectivity length 1.} MPS:
\begin{definition}[Injective MPS]
    Let $\mathcal{H}$ be a finite-dimensional Hilbert space, and $D\in \mathbb{N}$. An MPS tensor $A\in \mathcal{M}_D \otimes \mathcal{H}$ is called injective if the map $X\mapsto \mathfrak{M}_k(A,X)$ is injective for some $k\in\mathbb{N}$. The minimal such $k$ is called the injectivity length of $A$.
\end{definition}

It is easy to see that if the map $X\mapsto \mathfrak{M}_k(A,X)$ is injective for some $k\in\mathbb{N}$, then it is also injective for any $l>k$. The map $X\mapsto \mathfrak{M}_k(A,X)$ is injective if and only if $\Span\{A_{i_1} A_{i_2} \cdots A_{i_k}| i\in \{1,\dots, d\}^k\} = \mathcal{M}_D$.

Notice that for injective MPS, $\mathfrak{S}_n(A) \neq 0$, for any $n\in\mathbb{N}$. The precise statement characterizing the ground space of the parent Hamiltonian of such an MPS is the following:
\begin{theorem}\label{thm:parent_hamiltonian_uniqe_gs}
    Let $\mathcal{H}$ be a finite-dimensional Hilbert space, $D\in \mathbb{N}$, $A\in\mathcal{M}_D \otimes \mathcal{H}$ be an injective MPS tensor with injectivity length $k$, and $k<l\in \mathbb{N}$. Then the $l$-local parent Hamiltonian of $A$ exists and its ground space on $n\geq l$ particles is $\mathfrak{S}_n(A)$.
\end{theorem}

In particular, the ground space of the (open boundary) Hamiltonian is $D^2$-dimensional for all system sizes. In the following, we investigate an MPS that is injective with injectivity length $k=2$. \Cref{thm:parent_hamiltonian_uniqe_gs} guarantees that the ground space of its $l=3$-body parent Hamiltonian on $n$ particles is $\mathfrak{S}_n(A)$. In this paper, we consider the $l=2$-body parent Hamiltonian instead: the range of the Hamiltonian terms is the same as the injectivity length, $k=l$, and thus \Cref{thm:parent_hamiltonian_uniqe_gs} does not apply. We obtain our main result: the ground space of this parent Hamiltonian on $n$ particles grows with the system size, as it is of dimension $(n+1)^2$. 

\section{The $\SO(3)$ model}\label{sec:so_3}

Let a basis of $\so(3) \simeq \su(2)$ be $\{x,y,z\}$, with commutation relations
\begin{equation*}
    [x,y] = z, \quad  [y,z]=x, \quad [z,x] = y. 
\end{equation*}
The adjoint actions of $x,y,z$ can be represented by the following matrices relative to the basis $\{x,y,z\}$:
\begin{equation*}
    A_x = \begin{pmatrix}
         0 & 0 & 0 \\ 
         0 & 0 & -1 \\
         0 & 1 & 0
    \end{pmatrix}, \quad 
    A_y = \begin{pmatrix}
         0 & 0 & 1 \\ 
         0 & 0 & 0 \\
         -1 & 0 & 0
    \end{pmatrix}, \quad 
    A_z = \begin{pmatrix}
         0 & -1 & 0 \\ 
         1 & 0 & 0 \\
         0 & 0 & 0
    \end{pmatrix}. 
\end{equation*}

Consider the 4-dimensional complex Hilbert space $\mathcal{H}$ with orthonormal basis $B = \{x,y,z,e\}$, and the MPS tensor 
\begin{equation}\label{eq:MPS_tensor}
A = \sum_{i\in B} A_i \otimes \ket{i}\in \mathcal{M}_3 \otimes \mathcal{H},  
\end{equation}
where $A_e = \id$ is the identity matrix.
It is straightforward to check that this MPS is injective with injectivity length 2. Moreover, if we set 
\begin{equation*}
\psi(g) = \sum_{i,j\in\{x,y,z\}} \ad(g)_{i,j} \ket{i}\bra{j} \in \mathcal{B}(\mathcal{H}),
\qquad g\in \so(3),
\end{equation*}
then
\begin{align*}
    (\id_{\mathcal{M}_3} \otimes \psi(g))(A) &= \sum_{j\in B} A_j \otimes \psi(g) \ket{j} \\
    &= \sum_{i,j\in B} \mathrm{ad}(g)_{i,j} A_j \otimes \ket{i} \\
    &= - \sum_{i,j\in B} \mathrm{ad}(g)_{j,i} A_j \otimes \ket{i} \\
    &= - \sum_{i\in B} [A_g, A_i] \otimes \ket{i} \\
    &= -(\mathrm{ad}(A_g) \otimes \id_{\mathcal{H}})(A),
\end{align*}
for every $g\in \so(3)$. Graphically,
\begin{equation}
\label{eq:so3-symmetry}
    \begin{tikzpicture}[baseline=-1mm, xscale = 0.7, yscale = 0.5, font=\footnotesize]
        \draw (-1, 0) -- (1,0);
        \draw (0,0) -- (0,2);
        \node[tensor, label=right:$\psi(g)$] at (0,1) {};
        \node[tensor] at (0,0) {};
    \end{tikzpicture} \, = \;\,
    \begin{tikzpicture}[baseline=-1mm, xscale = 0.7, yscale = 0.5, font=\footnotesize]
        \draw (-1, 0) -- (2,0);
        \draw (0,0) -- (0,1);
        \node[tensor, label=below:$A_g$] at (1,0) {};
        \node[tensor] at (0,0) {};
    \end{tikzpicture} \; - \;
    \begin{tikzpicture}[baseline=-1mm, xscale = 0.7, yscale = 0.5, font=\footnotesize]
        \draw (-1, 0) -- (2,0);
        \draw (1,0) -- (1,1);
        \node[tensor, label=below:$A_g$] at (0,0) {};
        \node[tensor] at (1,0) {};
    \end{tikzpicture}\,. 
\end{equation}
This means that $\mathfrak{S}_n$ is symmetric under the action of $\so(3)$, for any $n$. In fact, if we let
\begin{equation*}
\psi^{(n)}(g) = \sum_{i=0}^{n-1} \id^{\otimes i} \otimes \psi(g) \otimes \id^{\otimes (n-i-1)}, 
\end{equation*}
then we have
\begin{equation*}
    \psi^{(n)}(g) \mathfrak{M}_n(A,X) = \sum_{i=1}^n \sum_{k\in \{1, \dots, d\}^n} \tr(X A_{k_1} \cdots [A_{k_i}, A_g] \cdots A_{k_n} )\, \ket{k_1 \cdots k_n} = \mathfrak{M}_n(A,[A_g , X]),
\end{equation*}
or graphically,
\begin{equation*}
    \sum_{i=1}^n  \ 
    \begin{tikzpicture}[baseline=-1mm, xscale = 0.7, yscale = 0.5, font=\footnotesize]
        \draw (-2, 0) rectangle (4,-1.1);
        \node[tensor, label=below:$X$] at (-1,0) {};            
        \foreach \x in {0,1, 3}{
            \draw (\x,0) --++ (0,2);
            \node[tensor] at (\x,0) {};            
        }
        \node[tensor,label=right:$\psi(g)_i$] at (1,1) {};            
        \node[fill=white] at (2,0) {$\dots$};
    \end{tikzpicture}\  = \;
    \begin{tikzpicture}[baseline=-1mm, xscale = 0.7, yscale = 0.5, font=\footnotesize]
        \draw (-2.5, 0) rectangle (4,-1.1);
        \node[tensor, label=below:$A_gX - XA_g$] at (-1,0) {};            
        \foreach \x in {0,1, 3}{
            \draw (\x,0) --++ (0,1);
            \node[tensor] at (\x,0) {};            
        }
        \node[fill=white] at (2,0) {$\dots$};
    \end{tikzpicture}\  . 
\end{equation*}
This is equivalent to the MPS being symmetric under the action of $\SO(3)$: 
\begin{equation*}
    \begin{tikzpicture}[baseline=-1mm, xscale = 0.9, yscale = 0.4, font=\footnotesize]
        \draw (-2, 0) rectangle (4,-1.375);
        \node[tensor, label=below:$X$] at (-1,0) {};            
        \foreach \x in {0,1, 3}{
            \draw (\x,0) --++ (0,2);
            \node[tensor] at (\x,0) {};            
            \node[tensor,label={[label distance=-1mm]right:$\Psi(a)$}] at (\x,1) {};            
        }
        \node[fill=white] at (2,0) {$\dots$};
    \end{tikzpicture}\  = \;
    \begin{tikzpicture}[baseline=-1mm, xscale = 0.7, yscale = 0.5, font=\footnotesize]
        \draw (-2.5, 0) rectangle (4,-1.1);
        \node[tensor, label=below:$\Phi(a)X\Phi(a)^{-1}$] at (-1,0) {};            
        \foreach \x in {0,1, 3}{
            \draw (\x,0) --++ (0,1);
            \node[tensor] at (\x,0) {};            
        }
        \node[fill=white] at (2,0) {$\dots$};
    \end{tikzpicture}\  ,
\end{equation*}
for all $a\in \SO(3)$, where $\Psi(\exp g) = \exp\psi(g)$ for $g\in \so(3)$, and $\Phi(\exp g) = \exp A_g$ is the defining representation.

The MPS is injective after blocking 2 sites, and thus its 3-body parent Hamiltonian has only the MPS as its ground space. Here we investigate a shorter-range Hamiltonian: As $\dim \mathcal{H}^{\otimes 2} = 4^2>3^2 \geq \dim\mathfrak{S}_2$, the MPS has a 2-body parent Hamiltonian. Let $H$ denote this parent Hamiltonian, and let $G_n$ be its ground space. As $\mathfrak{S}_2$ is symmetric under $\so(3)$ (and thus under $\SO(3)$ as well), the Hamiltonian $H$ and the ground space $G_n$ are symmetric as well. Our main result is the characterization of the ground space $G_n$ of this Hamiltonian:
\begin{theorem}\label{thm:main}
    The ground space $G_n$ contains each odd-dimensional irreducible representation of $\so(3)$ with dimension $\leq 2n+1$ exactly once, and in particular, its dimension is $\dim(G_n) =(n+1)^2$.
\end{theorem}

\begin{proof}
By definition, the ground space is
\begin{equation*}
    G_n = \bigcap_{i=0}^{n-2} \mathcal{H}^{\otimes i} \otimes \mathfrak{S}_2 \otimes \mathcal{H}^{\otimes (n-i-2)}.
\end{equation*}
Let $\mathfrak{S}_2^\perp$ be the orthocomplement of $\mathfrak{S}_2$, and $G_n^\perp$ the orthocomplement of $G_n$. Then
\begin{equation*}
    G_n^\perp =  \sum_{i=0}^{n-2} \mathcal{H}^{\otimes i} \otimes \mathfrak{S}_2^\perp \otimes \mathcal{H}^{\otimes (n-i-2)} .
\end{equation*}
Notice now that 
\begin{equation*}
    \dim G_n = \dim \bigl( \mathcal{H}^{\otimes n} / G_n^\perp \bigr),
\end{equation*}
and thus to calculate $\dim G_n$, it is enough to construct a basis of $\mathcal{H}^{\otimes n} / G_n^\perp$ instead of constructing a basis of $G_n$.  As $\mathfrak{S}_2^\perp$ is also $\so(3)$-invariant, $G_n \cong \mathcal{H}^{\otimes n} / G_n^\perp$ not only as vector spaces, but also as $\so(3)$-modules. 
Hence, investigating $\mathcal{H}^{\otimes n} / G_n^\perp$ will tell us how to decompose the ground space into a direct sum of irreducible representations of $\so(3)$.

Let us construct now $\mathfrak{S}_2^\perp$. For that, let us recall the definition of $\mathfrak{S}_2$:
\begin{equation}\label{eq:MPS_obc_space_2}
    \mathfrak{S}_2 = \Bigl\{\sum_{g,h\in B} \tr\left( X A_g A_h \right) \ket{gh} \,\Big|\, X\in \mathcal{M}_D\Bigr\},
\end{equation}
where $D=3$ in the case of $\so(3)$.
By definition,
\begin{equation}\label{eq:quadr_rel}
    \sum_{g,h\in B} \lambda_{gh} \ket{gh} \in \mathfrak{S}_2^\perp 
    \quad \Leftrightarrow \quad 
    \tr\left( X \sum_{g,h\in B} \bar{\lambda}_{gh} A_g A_h \right) =0, \ \forall X\in\mathcal{M}_D \quad \Leftrightarrow \quad
    \sum_{g,h\in B} \bar{\lambda}_{gh} A_g A_h = 0.
\end{equation}
As mentioned above, the MPS is injective with injectivity length 2, or in other words, $\dim \mathfrak{S}_2 = 9$. From this, it is immediate that $\dim \mathfrak{S}_2^\perp = 7$. In what follows, we construct an explicit basis for $\mathfrak{S}_2^\perp$. 
Due to the commutation relations and the fact that $A_e = \id$, we have
\begin{equation}\label{eq:comm_rel}
    A_g A_h - A_h A_g - A_{[g,h]} A_e =0,
\end{equation}
for all $g,h\in\so(3)$. Thus, we can immediately find the following 6 linearly independent vectors in $\mathfrak{S}_2^\perp$:
\begin{align*}
    v_1 & =\ket{ex} - \ket{xe}, \\
    v_2 & = \ket{ey} - \ket{ye}, \\
    v_3 &= \ket{ez} - \ket{ze}, \\
    v_4 &= \ket{xy} - \ket{yx} - \ket{ze}, \\
    v_5 &= \ket{yz} - \ket{zy} - \ket{xe}, \\
    v_6 &= \ket{zx} - \ket{xz} - \ket{ye}.
\end{align*}
The last vector, $v_7$, describes the fact that in the 3D irreducible representation, the Casimir operator is proportional to the identity:
\begin{equation*}
    A_xA_x + A_yA_y + A_zA_z = -2 \id = -2 A_e A_e,
\end{equation*}
and thus
\begin{equation*}
    v_7 = \ket{xx} + \ket{yy} + \ket{zz} + 2 \ket{ee}.
\end{equation*}
We have constructed $\mathfrak{S}_2^\perp$ as $\mathfrak{S}_2^\perp = J_1 \oplus J_2$, where $J_1 = \Span\{v_i\,|\,1\le i\le6\}$ and $J_2 = \mathbb{C} v_7$. Both vector spaces are invariant under the adjoint action of $\so(3)$. Correspondingly, $G_n^\perp$ can be decomposed as $ G_n^\perp = I_1 + I_2$, where
\begin{equation}\label{eq:i1i2}
  I_1 = \sum_{i=0}^{n-2} \mathcal{H}^{\otimes i} \otimes J_1 \otimes \mathcal{H}^{\otimes (n-i-2)} \quad \text{and} \quad 
  I_2 = \sum_{i=0}^{n-2} \mathcal{H}^{\otimes i} \otimes J_2 \otimes \mathcal{H}^{\otimes (n-i-2)}.
\end{equation}
Both $I_1$ and $I_2$ are invariant under the action of $\so(3)$. We construct $\mathcal{H}^{\otimes n}/G_n^\perp$ in two steps: first we factor out $I_1$, then $I_1 + I_2$. 

Let us write $\mathcal{H} = \mathbb{C} e \oplus \Span\{x,y,z\}$. This decomposition defines a grading and thus a filtration on $\mathcal{H}^{\otimes n}$:  $\mathcal{H}_{0} \subseteq \mathcal{H}_1 \subseteq \cdots \subseteq \mathcal{H}_n = \mathcal{H}^{\otimes n}$, where $\mathcal{H}_i$ means that $\Span\{x,y,z\}$ appears at most $i$ times in the tensor product.  This filtration is compatible with the action of $\so(3)$ on $\mathcal{H}^{\otimes n}$ (i.e., each $\mathcal{H}_i$ is closed under the action). The spaces $\mathcal{H}_i$ form a filtration of the (complexified) tensor algebra of $\so(3)$.

Let us consider now $\mathcal{H}^{\otimes n}/I_1$. As $I_1$ is $\so(3)$-invariant, this space is a representation of $\so(3)$. The space $\mathcal{H}^{\otimes n}/I_1$ also inherits the filtration of $\mathcal{H}^{\otimes n}$, and the filtration is again invariant under $\so(3)$.  We notice that the result is a filtration of the universal enveloping algebra $U(\so(3))$ of $\so(3)$. 
By the Poincaré–Birkhoff–Witt theorem, the canonical symmetrization map $\sigma \colon S(\so(3)) \to U(\so(3))$, defined on monomials by
$$\sigma(y_1 \cdots y_k) = \frac{1}{k!} \sum_{\pi \in S_k} y_{\pi(1)} \cdots y_{\pi(k)},$$
provides a vector space isomorphism between the symmetric algebra $S(\so(3))$ and $U(\so(3))$. Moreover, because $\sigma$ intertwines the derivation action of $\so(3)$ on $S(\so(3))$ with the adjoint action on $U(\so(3))$, $\sigma$ is an isomorphism of $\so(3)$-modules. Explicitly, the following elements form a basis for $\mathcal{H}^{\otimes n}/I_1$:
\begin{align*}
    \mathrm{Sym}\Bigl(\ket{\underbrace{x\cdots x}_{i} \ \underbrace{y \cdots y}_j \ \underbrace{z\cdots z}_k \ \underbrace{e\cdots e}_{n-i-j-k}}\Bigr) + I_1 
\end{align*}
for all $i,j,k$ such that $i+j+k\leq n$, where $\mathrm{Sym}$ denotes the symmetrization map
\begin{equation}\label{eq:sym-map}
\mathrm{Sym}\bigl(\ket{y_1 \cdots y_n}\bigr) = \frac{1}{n!} \sum_{\pi \in S_n} \ket{y_{\pi(1)} \cdots y_{\pi(n)}}.
\end{equation}
Note that we also have another basis for $\mathcal{H}^{\otimes n}/I_1$ given by the ordered monomials
\begin{align*}
    a_{ijk} = \ket{\underbrace{x\cdots x}_{i} \ \underbrace{y \cdots y}_j \ \underbrace{z\cdots z}_k \ \underbrace{e\cdots e}_{n-i-j-k}} + I_1,
\end{align*}
because the relations $I_1$ impose the commutativity of $x,y,z,e$ up to terms of lower degree.
Hence, we can identify $\mathcal{H}^{\otimes n}/I_1$ as an $\so(3)$-module with the vector space of three-variable polynomials of degree $\leq n$, via the isomorphism $a_{ijk} \mapsto x^i y^j z^k$. 
The dimension of this space is $\binom{n+3} 3$.

Let us now factor out $I_2 + I_1$ from this space. As $I_2$ is $\so(3)$-invariant, the result is a filtered vector space that is compatible with the $\so(3)$ action. Under the isomorphism defined as above, $I_1 + I_2$ is the space of polynomials of at most degree $n$ that have the form $p(x,y,z) = (x^2 +y^2 + z^2+2) q(x,y,z)$, where $\deg q\le n-2$. The resulting space $F_n$ is isomorphic to the space of all three-variable \emph{harmonic} polynomials of degree at most $n$.  We have a filtration $F_0\subseteq F_1 \subseteq \dots \subseteq F_n = \mathcal{H}^{\otimes n}/(I_1 + I_2)$, where $\dim F_n = (n+1)^2$. Note that $F_{n}/F_{n-1}$ is an irreducible representation of $\so(3)$, isomorphic to the space of three-variable harmonic polynomials of degree $n$,  with dimension $2n+1$.
\end{proof}

\begin{remark} We can construct a basis for this space explicitly: replace $zz$ with $-2ee-xx-yy$ in $a_{ijk}$. We obtain the representatives
\begin{equation}\label{eq:bijcij}
    b_{ij} = \ket{\underbrace{x\cdots x}_{i} \ \underbrace{y \cdots y}_j \ \underbrace{e\cdots e}_{n-i-j}} + I_1 + I_2 \quad \text{and} \quad c_{ij}=\ket{\underbrace{x\cdots x}_{i} \ \underbrace{y \cdots y}_j \ z \ \underbrace{e\cdots e}_{n-i-j-1}} + I_1 +I_2.
\end{equation}
There are $\binom{n+2}{2} + \binom{n+1}{2} = (n+1)^2$ such vectors, and they are linearly independent.
\end{remark}

\section{The general model}\label{sec:general}

In this section, we generalize the above construction to arbitrary compact simple Lie algebras and analyze the ground-space structure and degeneracy of the resulting Hamiltonian. We obtain a concrete formula for the ground-space degeneracy for the models based on the adjoint representation.

Let $\mathfrak{g}$ be a compact simple (real) Lie algebra and let $\mathcal{H}$ be, as a vector space, the complexification of $\mathfrak{g} \oplus \mathbb{R}e$, where $e$ is a symbol such that $e \notin \mathfrak g$: fixing a basis $B$ of $\mathfrak{g}$,
\begin{equation*}
    \mathcal H = \bigg\{ \sum_{g\in B}\lambda_g \cdot g + \mu \cdot e \ | \ \lambda\in\mathbb{C}^{|B|}, \mu \in \mathbb{C} \bigg\}.
\end{equation*}
Let $\kappa\colon \mathfrak{g} \times \mathfrak{g} \to \mathbb{C}$ be the Killing form of $\mathfrak{g}$. Recall that $\kappa$ is a negative-definite bilinear form on the real vector space $\mathfrak{g}$ that is invariant, i.e., for all $g,h,k\in \mathfrak{g}$,
\begin{equation*}
    \kappa([g,h],k) = \kappa(g,[h,k]).
\end{equation*}
Using $\kappa$, we can define a scalar product on $\mathcal{H}$: if $B$ is an orthonormal basis of $\mathfrak{g}$ with respect to $-\kappa$, we define the scalar product by letting $B\cup e$ be an ONB of $\mathcal{H}$. With this scalar product, $\mathcal{H}$ is a Hilbert space.  The Lie algebra $\mathfrak{g}$ acts on this Hilbert space via the adjoint representation $\psi\colon \mathfrak{g}\to \End(\mathcal{H})$ defined by 
\begin{equation*}
    \psi(g)(\lambda h + \nu e) := \lambda [g,h].
\end{equation*}
Note that the matrix $\psi(g)$ is anti-Hermitian for all $g\in \mathfrak{g}$:
\begin{align*}
    \langle \lambda h + c\,|\, \psi(g) (\mu k + d) \rangle &= \langle \lambda h + c \,|\, \mu [g,k] \rangle \\
    &= - \bar \lambda \mu \cdot \kappa( h, [g,k]) = - \bar \lambda \mu \cdot \kappa( [h,g], k) \\
    &= - \langle\psi(g)(\lambda h +c)\,|\, \mu k + d \rangle .
\end{align*}

Let $V$ be a complex vector space, $\phi\colon \mathfrak{g}\to \End(V)$ be a representation of $\mathfrak{g}$, and $B_{\mathfrak{g}}\subseteq \mathfrak{g}$ be an orthonormal basis (with respect to the negative Killing form) of $\mathfrak{g}$. The set $B = B_{\mathfrak{g}}\cup \{e\}$ then forms an orthonormal basis of $\mathcal{H}$. Let us set the MPS tensor $A \in \End(V) \otimes \mathcal{H}$ as 
\begin{equation}\label{eq:general_A}
    A = \id_V \otimes e + \sum_{g\in B_{\mathfrak{g}}} \phi(g) \otimes g .     
\end{equation}
Let $\psi^{(n)}\colon \mathfrak{g} \to \End(\mathcal{H}^{\otimes n})$ denote the $n$-fold tensor product of the representation $\psi$ of $\mathfrak{g}$, that is, 
\begin{equation*}
    \psi^{(n)}\colon g \mapsto \sum_{i=0}^{n-1} \id_{\mathcal{H}}^{\otimes i} \otimes \psi(g) \otimes \id_{\mathcal{H}}^{\otimes (n-i-1)}.
\end{equation*}
Then the MPS defined by the MPS tensor $A$ is symmetric under the action of $\mathfrak{g}$:
\begin{proposition}
    Let the MPS tensor $A$ be as in \cref{eq:general_A}. The MPS space $\mathfrak{S}_n(A)\subset \mathcal{H}^{\otimes n}$ is invariant under $\mathfrak{g}$ via the representation $\psi^{(n)}$: for all $h\in\mathfrak{g}$, we have 
    \begin{equation*}\label{eq:GS_covaiance}
        \psi^{(n)}(h) \cdot\mathfrak{M}_n(A,X) = \mathfrak{M}_n(A, [\phi(h),X]).
    \end{equation*}
\end{proposition}
\begin{proof}
    By orthonormality of the basis $B$, we can write $k = \sum_{g\in B} \langle g|k\rangle \cdot g$ for any $k\in\mathfrak{g}$, thus  
    \begin{equation*}
        \sum_{g\in B} g \cdot  \langle \psi(h) (g)|k \rangle = \sum_{g\in B} g \cdot  \scalprod{[h,g]}{k} = - \sum_{g\in B} g \cdot  \scalprod{[g,h]}{k} = - \sum_{g\in B} g \cdot  \scalprod{g}{[h,k]}  = [k,h].
    \end{equation*}
    That is,
    \begin{equation*}
        \sum_{g\in B} g \otimes \psi(h) (g) = \sum_{g\in B} [g,h] \otimes g.
    \end{equation*}
    Hence, as $\psi(h)(e) = 0$, the MPS tensor $A$ satisfies 
    \begin{equation}\label{eq:GS_covariance_1}
        (\id\otimes \psi(h))(A) = \sum_{g\in B} \phi(g) \otimes \psi(h) (g) = \sum_{g\in B} [\phi(g),\phi(h)] \otimes g,
    \end{equation}
    where we set $\phi(e) = \mathbbm{1}$. We can thus write 
    \begin{equation*}
        \begin{aligned}
            \psi^{(n)}(h) \cdot\mathfrak{M}_n(A,X) = &\sum_j \sum_{i_1 \dots i_n} \tr\big( X \phi(g_{i_1}) \dots \phi(g_{i_n}) \big) \cdot g_{i_1}\otimes \dots \otimes \ad_h(g_{i_j})\otimes  \dots \otimes g_{i_n} \\ 
            =&\sum_{i_1 \dots i_n} \sum_j \tr\big( X \phi(g_{i_1}) \dots (\phi(g_{i_j}) \phi(h) - \phi(h) \phi(g_{i_j}) ) \dots \phi(g_{i_n})\big) \cdot g_{i_1}\otimes   \dots \otimes g_{i_n}.
        \end{aligned}
    \end{equation*}
    Notice now that this is a telescopic sum, and thus 
    \begin{equation*}
        \begin{aligned}
            \psi^{(n)}(h) \cdot\mathfrak{M}_n(A,X)  
            &=\sum_{i_1 \dots i_n} -\tr\big( X \phi(h) \phi(g_{i_1}) \dots  \phi(g_{i_n})\big) \cdot g_{i_1}\otimes   \dots \otimes g_{i_n} \\
            &+\sum_{i_1 \dots i_n} \tr\big( X \phi(g_{i_1}) \dots  \phi(g_{i_n}) \phi(h) \big) \cdot g_{i_1}\otimes   \dots \otimes g_{i_n} \\ &=  \mathfrak{M}_n(A,[\phi(h),X]).
        \end{aligned}
    \end{equation*}
\end{proof}

This means that the MPS, and thus its parent Hamiltonian as well, is invariant under the tensor product of the local unitary operators $u_g = \exp \psi(g)$, for all $g\in \mathfrak{g}$:
\begin{equation*}
    u_g^{\otimes n} \cdot \mathfrak{S}_n (A) = \mathfrak{S}_n (A).
\end{equation*}
We can generalize our main result, \Cref{thm:main}, to this scenario and characterize the ground space of the parent Hamiltonian of this MPS:  
\begin{theorem}
Let an MPS tensor $A$ be as in \cref{eq:general_A}, where $\mathfrak{g}$ is a simple compact Lie algebra  
and $\phi$ is the adjoint representation of $\mathfrak{g}$. Then the ground space $G_n$ of its $2$-local Hamiltonian on $n$ particles 
is isomorphic to the space of harmonic polynomials of degree $\le n$ in $d=\dim \mathfrak{g}$ variables,
and has dimension $\frac{d+2n-1}{d-1} \binom{d+n-2}{d-2}$.
\end{theorem}

\begin{proof}
Let us introduce the matrices $A_g=\phi(g)=\ad_g$ for $g\in\mathfrak{g}$ and $A_e=\id_{\mathfrak{g}}$. This allows us to rewrite \cref{eq:general_A} more compactly as
\begin{equation}\label{eq:general_Ae}
    A = \sum_{g\in B} A_g \otimes g .     
\end{equation}
Recall the definition of $\mathfrak{S}_2$ from \eqref{eq:MPS_obc_space_2},
\begin{equation*}
    \mathfrak{S}_2 = \Bigl\{\sum_{g,h\in B} \tr\left( X A_g A_h \right) \ket{gh} \,\Big|\, X\in \mathcal{M}_d\Bigr\}.
\end{equation*}
We note that elements of $\mathfrak{S}_2^\perp$ are in bijection with quadratic relations satisfied by $A_g$ (cf.\ \eqref{eq:quadr_rel}):
\begin{equation}\label{eq:quadr_rel2}
\begin{split}
    \sum_{g,h\in B} \lambda_{gh} \ket{g} \otimes \ket{h} \in \mathfrak{S}_2^\perp \quad &\Leftrightarrow \quad 
    \tr\Bigl( X \sum_{g,h\in B} \bar\lambda_{g,h} A_g A_h \Bigr) =0, \ \forall X\in\mathcal{M}_d \\ 
    &\Leftrightarrow \quad
    \sum_{g,h\in B} \bar{\lambda}_{gh} A_g A_h = 0,
\end{split}
\end{equation}
because the trace form $\tr(XY)$ is non-degenerate for $X,Y\in\mathcal{M}_d$. 

Due to the commutation relations and the fact that $A_e = \id$, we have for $g,h\in\mathfrak{g}$:
\begin{align}\label{eq:comm_rel2}
    A_g A_h - A_h A_g - A_{[g,h]} A_e &= 0, \\ \label{eq:comm_rel3}
    A_g A_e - A_e A_g &= 0.
\end{align}
Another relation comes from the quadratic Casimir element: since $\mathfrak{g}$ is simple, its adjoint representation is irreducible. Hence, by Schur's Lemma, the Casimir element acts in it as a scalar multiple of $\id$, that is,
\begin{equation}\label{eq:casimir_repr}
    \sum_{g\in B_\mathfrak{g}} A_g A_g = \lambda\cdot  \id.
\end{equation}
To determine the value of $\lambda$, recall that $B_{\mathfrak g} = B\cap \mathfrak{g}$ is an orthonormal basis for $\mathfrak{g}$ relative to the negative Killing form; hence
\begin{equation}\label{eq:killing}
\tr(A_g A_h) = -\delta_{g,h}, \qquad g,h\in B_{\mathfrak{g}}.
\end{equation}
Taking the trace in \cref{eq:casimir_repr}, we find that $\lambda = -1$ due to the normalization~\eqref{eq:killing}.  Then we have
\begin{equation}\label{eq:quadr_rel3}
    \sum_{g\in B_{\mathfrak{g}}} A_g A_g = -\id = -A_e A_e,
\end{equation}
or equivalently,
\begin{equation}\label{eq:quadr_rel3e}
    \sum_{g\in B} A_g A_g = 0.
\end{equation}

According to \eqref{eq:quadr_rel2}, the relations \eqref{eq:comm_rel2}, \eqref{eq:comm_rel3}, \eqref{eq:quadr_rel3e} correspond to the following elements of $\mathfrak{S}_2^\perp$, respectively:
\begin{align}\label{eq:v_rel2}
    v_{g,h} &:= \ket{g} \otimes \ket{h} - \ket{h} \otimes \ket{g} - \ket{[g,h]} \otimes \ket{e} \in \mathfrak{S}_2^\perp, \\ \label{eq:v_rel3}
    v_{g} &:= \ket{g} \otimes \ket{e} - \ket{e} \otimes \ket{g} \in \mathfrak{S}_2^\perp, \\ \label{eq:v_quadr}
    v &:= \sum_{g\in B} \ket{g} \otimes \ket{g} \in \mathfrak{S}_2^\perp.
\end{align}
As in the proof of \Cref{thm:main}, let us set 
\begin{equation*}
  J_1 = \Span\{ v_{g,h} \,|\, g,h\in B_{\mathfrak{g}},\,g\ne h \} \oplus \Span\{ v_{g} \,|\, g\in B_{\mathfrak{g}} \}, \quad \text{and} \quad J_2 = \Span\{v\}. 
\end{equation*}

Next, we will prove that $\mathfrak{S}_2^\perp  = J_1 \oplus J_2$.  For that, we will show that $\dim \mathfrak{S}_2 \geq \binom{d}{2} +2d$, and thus $\dim \mathfrak{S}_2^\perp \leq (d+1)^2 - \binom{d}{2} - 2d = \binom{d}{2} + d + 1$. As the dimension of $J_1 \oplus J_2$  is already $\binom{d}{2} + d + 1$, this implies that $\mathfrak{S}_2^\perp  = J_1 \oplus J_2$. To prove that $\dim \mathfrak{S}_2 \geq \binom{d}{2} +2d$, let us define the following $\binom{d}{2} +2d$ vectors in $\mathfrak{S}_2$ for $g,h\in B_{\mathfrak{g}}$:
\begin{align}
    w_g &= \sum_{k,l\in B} \tr (A_g A_k A_l) \ket{kl}, \label{eq:wg} \\ 
    w_{g,h} &= \sum_{k,l\in B} \bigl(\bra{g} A_k A_l \ket{h} + \bra{h} A_k A_l \ket{g}\bigr) \ket{kl}. \label{eq:wgh}
\end{align}
Note that for $v_{g,h}$ we consider all unordered pairs of distinct $g,h\in B_{\mathfrak{g}}$ (hence there are $\binom{d}{2}$ such vectors), while 
for $w_{g,h}$ we consider all unordered pairs of $g,h\in B_{\mathfrak{g}}$ where $g=h$ is allowed (and there are $\binom{d}{2}+d$ such vectors). We will show now that these vectors are all linearly independent. 

First of all, by \cref{eq:killing} and $A_g^T = -A_g$, we have:
\begin{align*}
    w_g &=  \sum_{k\in B_{\mathfrak{g}}} \tr (A_g A_k) \ket{ke} +  \sum_{l\in B_{\mathfrak{g}}} \tr (A_g A_l) \ket{el} + \frac{1}{2} \sum_{k,l\in B_{\mathfrak{g}}} \tr \bigg( A_g A_k A_l + (A_g A_k A_l)^T \bigg) \ket{kl} \\
     &= -\ket{ge} -  \ket{eg} -  \frac{1}{2}\sum_{k,l\in B_{\mathfrak{g}}} \tr \bigl(A_k [A_g ,A_l] \bigr) \ket{kl} \\
     &= -\ket{ge} -  \ket{eg} + \frac{1}{2}\sum_{l\in B_{\mathfrak{g}}} \ket{[g,l]} \otimes \ket{l}.
\end{align*}
A similar calculation applies to $w_{g,h}$:
\begin{align*}
w_{g,h} &= 2\delta_{g,h} \ket{ee} + \sum_{k\in B_{\mathfrak{g}}} \bigl(\bra{g} A_k \ket{h} + \bra{h} A_k \ket{g}\bigr) \ket{ke} 
+ \sum_{l\in B_{\mathfrak{g}}} \bigl(\bra{g} A_l \ket{h} + \bra{h} A_l \ket{g}\bigr) \ket{el} \\
&+\sum_{k,l\in B_{\mathfrak{g}}} \bigl(\bra{g} A_k A_l \ket{h} + \bra{h} A_k A_l \ket{g}\bigr) \ket{kl} \\
&= 2\delta_{g,h} \ket{ee} + \sum_{k,l\in B_{\mathfrak{g}}} \bigl(\bra{g} A_k A_l \ket{h} + \bra{h} A_k A_l \ket{g}\bigr) \ket{kl},
\end{align*}
where we used that 
$$\bra{g} A_k \ket{h} = \scalprod{g}{[k,h]} = \scalprod{[h,g]}{k},$$
which implies $\bra{g} A_k \ket{h} + \bra{h} A_k \ket{g} = 0$.
Next, using $A_k = \ad_k$ for $k\in\mathfrak{g}$ and $\bra{g} A_k A_l \ket{h} = -\scalprod{[k,g]}{[l,h]}$, we rewrite \cref{eq:wgh} as
$$w_{g,h} = 2\delta_{g,h} \ket{ee} - \sum_{k,l \in B_\mathfrak{g}} \bigl( \scalprod{[k,g]}{[l,h]} + \scalprod{[k,h]}{[l,g]} \bigr) \ket{kl}.$$
Recognizing $\scalprod{[k,g]}{[l,h]} = -\scalprod{k}{[g,[h,l]]}$, and expanding in the orthonormal basis $B_\mathfrak{g}$, we obtain
\begin{align*}
w_{g,h} - 2\delta_{g,h} \ket{ee} 
&= \sum_{l\in B_{\mathfrak{g}}} \Bigl( \ket{[g,[h,l]]} \otimes \ket{l}+ \ket{[h,[g,l]]} \otimes \ket{l} \Bigr) \\
&= -\sum_{l\in B_{\mathfrak{g}}} \Bigl( \ket{[h,l]} \otimes \ket{[g,l]} + \ket{[g,l]} \otimes \ket{[h,l]} \Bigr).
\end{align*}
For the last equality, we used that the split Casimir element
$$C = \sum_{l\in B_{\mathfrak{g}}} \ket{l} \otimes \ket{l} \in \mathfrak{g} \otimes \mathfrak{g}$$
is invariant under the adjoint action of $\mathfrak{g}$.
Therefore,
$$w_{g,h} = 2\delta_{g,h} \ket{ee} 
- \sum_{l \in B_\mathfrak{g}} \bigl(\ad_l \otimes \ad_l\bigr) \bigl( \ket{g} \otimes \ket{h} + \ket{h} \otimes \ket{g} \bigr).$$

Clearly, the set $\{w_g\}_{g\in B_{\mathfrak{g}}}$ is linearly independent, and as $w_{g,h} - 2\delta_{g,h} \ket{ee} \in \mathfrak{g} \otimes \mathfrak{g}$, it is also independent from  $\{w_{g,h}\}_{g,h\in B_{\mathfrak{g}}}$. Thus, we are only left to prove that $w_{g,h}$ are linearly independent, for which it suffices to show that the module homomorphism 
$$F\colon S^2 \mathfrak{g} \to S^2 \mathfrak{g}, \qquad F = \sum_{l\in B_{\mathfrak{g}}} \ad_{l} \otimes \ad_{l},$$ 
is injective. 

Our key observation is that $F$ is the adjoint action of the split Casimir element $C$ on $S^2 \mathfrak{g}$,
whose eigenvalues are determined by Vogel's universal parameters \cite{vogel1999} (see also \cite{isaev2026vogel}). Vogel demonstrated that for any simple finite-dimensional Lie algebra, the symmetric square of the adjoint representation decomposes into at most four sub-representations:
\begin{equation}
S^2\mathfrak{g} = \mathbb{C} \oplus Y_2(\alpha) \oplus Y_2(\beta) \oplus Y_2(\gamma),
\end{equation}
where $\mathbb{C}$ represents the trivial one-dimensional representation. By Schur's lemma, $F$ acts as a scalar on each summand. Up to a nonzero normalization factor, the eigenvalues of 
$F$  
on the respective components $\mathbb{C}, Y_2(\alpha), Y_2(\beta), Y_2(\gamma)$ are given by $-2t, -\alpha, -\beta, -\gamma$, where 
\begin{equation}
t = \alpha + \beta + \gamma > 0.
\end{equation}
The projectively defined numbers $(\alpha, \beta, \gamma)$, known as the Vogel parameters, together with $t$, are non-zero for all simple Lie algebras $\mathfrak{g}$, which proves that $F$ is injective. Below we list the known values of the Vogel parameters \cite{vogel1999}: 
\begin{enumerate}
\item \textbf{Type $A_r$ ($\mathfrak{sl}_{r+1}$):}
The parameters can be scaled to $\{-2, 2, r+1\}$. Here, $t = r+1$. The eigenvalues are $-2(r+1), 2, -2, -(r+1)$. Since $r \ge 1$, none of these values is zero.
\item \textbf{Types $B_r$ and $D_r$ ($\mathfrak{so}_{N}$):}
The parameters are $\{-2, 4, N-4\}$. For $\mathfrak{so}_5$ ($B_2$) and higher, $N \ge 5$, meaning $N-4 \neq 0$. Thus, all parameters are non-zero.
\item \textbf{Type $C_r$ ($\mathfrak{sp}_{2r}$):}
The parameters are $\{-2, 1, r+2\}$. For $r \ge 1$, all elements are strictly non-zero.
\item \textbf{Exceptional Lie Algebras ($\mathfrak{g}_2, \mathfrak{f}_4, \mathfrak{e}_6, \mathfrak{e}_7, \mathfrak{e}_8$):}
The parameters are given by specific non-zero rational fractions. For instance, for $\mathfrak{g}_2$, the parameters are $\{-2, 10/3, 8/3\}$. None of the parameters vanish for any exceptional algebra.
\end{enumerate}

We have thus seen that $\mathfrak{S}_2^\perp = G_2^\perp = J_1 \oplus J_2$, and thus, in general, $G_n^\perp$ can be decomposed as $ G_n^\perp = I_1 + I_2$, where $I_1$ and $I_2$ are defined as in \eqref{eq:i1i2},
\begin{equation*}
  I_1 = \sum_{i=0}^{n-2} \mathcal{H}^{\otimes i} \otimes J_1 \otimes \mathcal{H}^{\otimes (n-i-2)} \quad \text{and} \quad 
  I_2 = \sum_{i=0}^{n-2} \mathcal{H}^{\otimes i} \otimes J_2 \otimes \mathcal{H}^{\otimes (n-i-2)}.
\end{equation*}
By the Poincar\'e--Birkhoff--Witt theorem (see, e.g., \cite{dixmier1996enveloping}), the space $\mathcal{H}^{\otimes n}/I_1$ can be identified as a representation of $\mathfrak{g}$ with the space of polynomials from $S(\mathfrak{g}) \cong \mathbb{C}[x_1,\dots,x_d]$ of degree $\le n$, where $B_{\mathfrak{g}}=\{x_1,\dots,x_d\}$.

It follows that we can identify the ground space $G_n$ with the space of polynomials in $x_1,\dots,x_d$ of degree $\le n$ (which is linearly spanned by \emph{homogeneous} polynomials of degree $\le n$), modulo the relation $\sum x_i^2 = -1$. Since every homogeneous polynomial of degree $m$ can be written uniquely in the form
\[ \sum_{0\le k \le m/2} h_k(x_1,\dots,x_d) \cdot \Bigl(\sum_{i=1}^d x_i^2\Bigr)^k, \]
where $h_k$ is a homogeneous harmonic polynomial of degree $m-2k$,
we see that $G_n$ is isomorphic to the space of $d$-variable harmonic polynomials of degree $\le n$. The dimension of $d$-variable homogeneous harmonic polynomials of degree $m$ is $\binom{d + m -1}{d-1} - \binom{d+m -3}{d-1}$.
Hence, the degeneracy of the ground space is 
    \begin{align*}
        \dim G_n &= \sum_{m=0}^{n} \Bigl[\binom{d + m -1}{d-1} - \binom{d+m -3}{d-1} \Bigr]\\
        &= \binom{d + n -1}{d-1} + \binom{d+n -2}{d-1}  \\
        &= \frac{d+2n-1}{d-1} \binom{d+n-2}{d-2}.
    \end{align*}
This completes the proof of the theorem.
\end{proof}

\section{Periodic boundary conditions in the $\SO(3)$ model} \label{sec:pbc}

In this section, we analyze the ground-space degeneracy of the $\SO(3)$ model with periodic boundary conditions. The $\SO(3)$ model with periodic boundary conditions is given by the Hamiltonian 
\begin{equation}
    H^\mathrm{PBC} = H + h_{n,1}
\end{equation}
where $H$ is given in \cref{parentH} and $h_{n,1}$ is the local Hamiltonian term acting on the last and first particle. As $h_{n,1}$ is positive semidefinite, this means that the ground space $G_n^\mathrm{PBC}$ of $H^\mathrm{PBC}$ satisfies $G_n^\mathrm{PBC} \subseteq G_n$. As above, we can find the dimension of the space by $\dim G_n^\mathrm{PBC} = \dim \big(\mathcal{H}^{\otimes n} / (G_n^\mathrm{PBC})^\perp\big)$. Here, $(G_n^\mathrm{PBC})^\perp = I_1 + I_2 + I_3$, where $I_1$ and $I_2$ are given in \cref{eq:i1i2}, and $I_3$ is composed of states that are locally orthogonal to the ground space on the last and first particle: 
\begin{align*}
    I_3 = \Span \Bigl\{& \ket{e a_1 \cdots a_{n-2} u} - \ket{ua_1 \cdots a_{n-2} e} \,\big|\, a_1, \dots, a_{n-2} \in \mathcal{H}, u\in \{x,y,z\} \Bigr\} \\
    + \,\Span \Bigl\{& \ket{y a_1 \cdots a_{n-2} x} - \ket{xa_1 \cdots a_{n-2} y } - \ket{e a_1 \cdots a_{n-2} z}, \\
&  \ket{z a_1 \cdots a_{n-2} y} - \ket{ya_1 \cdots a_{n-2} z } - \ket{e a_1 \cdots a_{n-2} x}, \\
&   \ket{x a_1 \cdots a_{n-2} z} - \ket{za_1 \cdots a_{n-2} x } - \ket{e a_1 \cdots a_{n-2} y}, \\
&   \ket{x a_1 \cdots a_{n-2} x} + \ket{ya_1 \cdots a_{n-2} y } + \ket{z a_1 \cdots a_{n-2} z} + 2\ket{e a_1 \cdots a_{n-2} e} \,\big|\, a_1, \dots, a_{n-2} \in \mathcal{H} \Bigr\}.
\end{align*}
We can thus construct $\mathcal{H}^{\otimes n} / (G_n^\mathrm{PBC})^\perp$ as $\bigl(\mathcal{H}^{\otimes n}/ (I_1 +I_2)\bigr)\big/\bigl((I_3 + I_1 + I_2)/(I_1+I_2)\bigr)$. We have already constructed $\mathcal{H}^{\otimes n}/ (I_1 +I_2)$, and in the following proof we factor out $I_3$ from it.

\begin{theorem}
    \label{thm:pbc-gs-structure}
    The ground-space degeneracy of $H^\mathrm{PBC}$ is $2n+2$, and it consists of the $1$-dimensional trivial representation plus the irreducible representation of dimension $2n+1$.
\end{theorem}

\begin{proof}
    Recall from the proof of \Cref{thm:main} (cf.\ \cref{eq:bijcij}), that the open-boundary condition ground space is
    \begin{equation*}
        G_{n}^{\mathrm{OBC}} = \bigoplus_{i=0}^n V_i,
    \end{equation*}
    where $V_i$ are spin-$i$ irreducible representations. As $G_{n}^{\mathrm{PBC}} \subseteq G_{n}^{\mathrm{OBC}}$ and it is also $\SO(3)$-invariant, we have
    \begin{equation*}
        G_{n}^{\mathrm{PBC}} = \bigoplus_{i\in \mathcal{K}} V_i,
    \end{equation*}
    for some subset $\mathcal{K} \subseteq \{0,\dots, n\}$. We will show that $\mathcal{K} = \{0,n\}$. 

    Notice that the highest weight vector of $V_k$ is
    \begin{equation*}
        u_k := \ket{\underbrace{u\cdots u}_k e \cdots e } + I_1 + I_2,
    \end{equation*}
    where $u = x+ i y$.
    We will prove that, for $1 \le k\le n-1$, $u_k$ becomes $0$ after factoring out $I_3$, i.e., that $u_k\in I_1 + I_2 + I_3$. Since each irreducible representation is either entirely contained in $G_{n}^{\mathrm{PBC}}$ or entirely absent from it, this implies that $G_{n}^{\mathrm{PBC}} \subseteq V_0 + V_n$. 

    To show that $u_k \equiv 0 \mod {I_1 + I_2 + I_3}$, we first  use that $i\ket{eu} \equiv \ket{uz}-\ket{zu}$ on the last and first particle, arriving at
    \begin{align*}
        i u_{k} &= i \ket{u\underbrace{u \cdots u}_{k-1} e\cdots ee} \\ 
        &\equiv \ket{z\underbrace{u \cdots u}_{k-1} e\cdots eu} - \ket{u\underbrace{u \cdots u}_{k-1} e\cdots ez} \\
        &\equiv \ket{z\underbrace{u \cdots u}_{k} e\cdots e} - \ket{\underbrace{u \cdots u}_{k} ze\cdots e} .       
    \end{align*}
    Then in the first term, we commute $z$ through all $u$'s. That is, we write the telescopic sum
    \begin{align*}
        i u_{k} \equiv \sum_{m=0}^{k-1} \Bigl(\ket{\underbrace{u \cdots u}_{m} zu \underbrace{u \cdots u}_{k-1-m} e\cdots e} - \ket{\underbrace{u \cdots u}_{m} uz\underbrace{u \cdots u}_{k-1-m} e\cdots e} \Bigr).    
    \end{align*}
    Using the commutation relations $\ket{zu}-\ket{uz} \equiv -i\ket{eu}$, we obtain
    \begin{align*}
        i u_{k} \equiv -i\sum_{m=0}^{k-1} \ket{\underbrace{u \cdots u}_{m} eu \underbrace{u \cdots u}_{k-1-m} e\cdots e}
        \equiv -ik u_k.
    \end{align*}
    Therefore, $u_{k} \equiv 0$, finishing the proof that $G_{n}^{\mathrm{PBC}} \subseteq V_0 + V_n$.

    It remains to show that $V_0\not\equiv 0$ and $V_n\not \equiv 0$. As noted above, it is enough to find a single state from both sectors that is non-zero, or equivalently, that is a ground state of the Hamiltonian $H_n^{\mathrm{PBC}}$. As the Hamiltonian is frustration free and $V_0$ and $V_n$ are in the ground space of $H_n^{\mathrm{OBC}}$, it is enough to show that there are translation invariant states in the two sectors. This is trivial because
    $u_0 = \ket{e\cdots e} \in V_0$  and $u_n = \ket{u\cdots u}\in V_n$.
\end{proof}

\section{Physical structure of the model}\label{sec:physical}

We conclude with a brief discussion of the physical structure of the model.

The physical Hilbert space of the model at each site
decomposes into a direct sum of a spin-$1$ and a spin-$0$ component, which can
be naturally interpreted as either a spin-$1$ particle or a hole
(akin to the Hilbert space of a Bose--Hubbard model made up of spin-$1$ particles).
This is checked most conveniently 
by working with the basis 
\begin{equation*}
\ket{S=1,S_z=0} \equiv \ket{1,0} := \ket{z}, \quad
\ket{S=1,S_z=\pm1} \equiv \ket{1,\pm1}:= \tfrac{1}{\sqrt{2}}(\mp\ket{x}-i\ket{y}), \quad
\ket{S=0} \equiv \ket{0} := \ket{e},
\end{equation*}
where the vectors with $S=1$ correspond to the spin-$1$ particle,  and $S=0$ corresponds to the hole of
the physical Hilbert space. In this basis, the (Hermitian) 
generator $i\psi(g)$ of the  physical  
$\mathfrak{so}(3)$ symmetry of the tensor,
Eq.~\eqref{eq:so3-symmetry}, acts as a spin-$0$ and a spin-$1$ irreducible
representation, with the spin-$1$ representation given by
\begin{equation*}
S_x = \frac{1}{\sqrt{2}}
    \begin{pmatrix} 0 & 1 & 0 \\ 1 & 0 & 1 \\ 0 & 1 & 0 \end{pmatrix}
    ,\quad
S_y = \frac{1}{\sqrt{2}}
    \begin{pmatrix} 0 & -i & 0 \\ i & 0 & -i \\ 0 & i & 0 \end{pmatrix}
    ,\quad
S_z = 
    \begin{pmatrix} 1 & 0 & 0 \\ 0 & 0 & 0 \\ 0 & 0 & -1 \end{pmatrix}
    ,
\end{equation*}
that is, the MPS is invariant under said symmetry action.
The resulting model thus has a local Hilbert space $\mathcal H$, 
which decomposes into a direct sum of a spin-0 and a spin-1 space,
$\mathcal H = \mathcal H_0\oplus \mathcal H_1$.

The  frustration-free and $\mathfrak{so}(3)$-invariant $2$-site parent Hamiltonian $h$ is,
up to a constant, equal to minus the projector onto  $\mathfrak S_2$, $h=\id-\Pi_{\mathfrak S_2}$.
To understand its physical structure, we first observe that the 2-site Hilbert space
$\mathcal H\otimes \mathcal H = (\mathcal H_0\oplus \mathcal H_1)\otimes (\mathcal
H_0\oplus \mathcal H_1)$ decomposes into a
sum of two spin-0, three spin-1, and one spin-2 irreducible representations of $\mathfrak{so}(3)$. Observe that: 
(i) $\mathfrak S_2$  is obtained by placing all boundary
conditions $X$ around two $A$ tensors, Eq.~\eqref{eq:MPS_obc_space_2}, 
(ii) the virtual degrees of freedom of $A$ transform as a spin-$1$ irrep, thus
the space of all boundary conditions $X$ transforms under the adjoint representation
$1\otimes 1=0\oplus 1\oplus 2$ and (e.g.\ by dimension counting) fully spans this space,
and (iii) the two contracted $A$ tensors in \eqref{eq:MPS_obc_space_2} form an
injective and $\mathfrak{so}(3)$-symmetric map from the virtual to the physical
system. It follows that the ground space $\mathfrak S_2$ of $h$ contains
precisely one spin-$0$, one spin-$1$, and one spin-$2$ irreducible
representation each. 

What is the structure of those different irrep spaces within the $2$-site ground space,
and thus of the different terms in the two-body parent Hamiltonian $h$? 
\begin{itemize}
\item
The first component of the ground space is the unique spin-$2$ component of the
$2$-site Hilbert space $\mathcal H\otimes \mathcal H$. It is spanned by all
states $\ket{\vec n}\otimes \ket{\vec n}$, where $\ket{\vec n}\in\mathcal H_1$
lives in the spin-$1$ space and has maximum spin $\vec n\cdot \vec S = \sum
n_\alpha S_\alpha$ in a given direction $\vec n \in \mathbb R^3$ ($|\vec n|=1$).
In particular, this implies that on $N$ sites, both for open and periodic
boundary conditions, the ferromagnetic states $\ket{\vec n}^{\otimes N}$  (which
precisely span the space with maximum total spin $N$) are ground states of the
periodic $N$-site chain.
\item
The spin-$0$ space in $\mathcal H\otimes \mathcal H$ is two-fold degenerate,
spanned by $\ket{0}\ket{0}=\ket{e}\ket{e}$ and 
\begin{equation}
\label{eq:phys:1x1-joint0}
\tfrac{1}{\sqrt{3}}\bigl(-\ket{1,+1}\ket{1,-1}+\ket{1,0}\ket{1,0}-\ket{1,-1}\ket{1,+1}\bigr)
 = \tfrac{1}{\sqrt{3}}\bigl(\ket{x}\ket{x}+\ket{y}\ket{y}+\ket{z}\ket{z}\bigr)\ .
\end{equation}
It contains the vector
$v_7=\ket{x}\ket{x}+\ket{y}\ket{y}+\ket{z}\ket{z}+2\ket{e}\ket{e}$, which is in
the excited space of the Hamiltonian. The ground space thus contains its
orthocomplement within the joint spin-$0$ space, namely, 
\[
\ket{0}\ket{0} + 
\tfrac{2}{3}\bigl(\ket{1,+1}\ket{1,-1}-\ket{1,0}\ket{1,0}+\ket{1,-1}\ket{1,+1}\bigr)\ .
\]
The action of the Hamiltonian in the spin-$0$ sector thus describes an
interaction where a pair of spin-$1$ particles in a joint singlet (i.e.,
spin-$0$) state,
Eq.~\eqref{eq:phys:1x1-joint0},
is created or annihilated, i.e., a pairing interaction. 
\item 
Finally, it follows that the spin-$1$ component in the ground space can be
obtained by the orthocomplement of $v_1,\dots,v_6$ in the $3$-fold degenerate
spin-$1$ component of two adjacent sites (obtained as $\mathcal H_1\otimes
\mathcal H_0$, $\mathcal H_0\otimes \mathcal H_1$,  i.e., one hole and one
spin-$1$ particle, together with the spin-$1$ fusion channel of $\mathcal
H_1\otimes \mathcal H_1$,  i.e., of two particles). The action of the Hamiltonian in the
spin-$1$ sector thus 
describes a spin-$1$ particle that hops to the adjacent site, which is in the
vacuum, while at the same time, a second spin-$1$ particle  ---
which forms a joint spin-$1$ state with the other particle ---
can be coherently created.
\end{itemize}

By construction, the periodic boundary condition Hamiltonian has the original
MPS with periodic boundary condition, i.e. $X = I$, as a ground state. In addition, as we have seen above, it has the
ferromagnetic states $\ket{\vec n}^{\otimes N}$, with $\ket{\vec n}$ a state with
maximal magnetization in direction $\vec n$, as ground states. These states span the
space with total spin $N$. Together, these states thus fully characterize the
ground space for periodic boundary conditions, cf.\ Theorem~\ref{thm:pbc-gs-structure}.

Given that the model has a ferromagnetic ground state, it is natural to expect
that it has gapless spin wave excitations above the ferromagnetic ground state.
This can be indeed verified through the spin wave ansatz
\begin{equation}
    \label{eq:spinwaveansatz}
\ket{\psi(k)}  = \frac{1}{\sqrt{N}}\sum_{\ell=1}^N e^{ik\ell} S_\ell^-\ket{1,+1}^{\otimes N}\ ,
\end{equation}
with $S^-=\tfrac{1}{\sqrt{2}}(S_x-iS_y)$, for which one obtains the variational
energy  $\bra{\psi(k)} H \ket{\psi(k)} = \tfrac45(1-\cos(k))$, where here $H$
is the PBC Hamiltonian. Since for $k\ne0$, the variational energy provides an
upper bound to the energy of the lowest-lying eigenstate with momentum $k$, this
rigorously establishes that the model is gapless, with an (at least)
quadratically closing gap. The existence of a spin wave band with a quadratic
dispersion can also be confirmed numerically through exact diagonalization, with
the numerically obtained energy roughly matching that of the spin wave ansatz
\eqref{eq:spinwaveansatz}.

\section{Conclusion and outlook}\label{sec:conclusion}

In this paper, we have introduced a systematic construction for certain matrix product states that possess continuous symmetries. Unlike the AKLT model, the parent Hamiltonians of these states have polynomially growing ground-space degeneracies in both the open- and periodic-boundary cases. Due to the symmetry of these models, the ground-space degeneracy can be understood as the presence of non-trivial irreducible representations in the ground space. We have shown how the ground space for the construction with $\so(3)$ symmetry decomposes into irreducible representations, both in the open and periodic boundary cases, and we have generalized these results to other simple Lie algebras.

For the $\so(3)$-symmetric model with open boundary conditions on $n$ sites, we have found that the irreducible representations that appear in the decomposition of the ground space are those that have odd dimension at most $2n+1$ (that is, integer spin $\le n$). Moreover, each of these irreducible representations appears with multiplicity one, resulting in a ground-space degeneracy that grows quadratically with system size. We have also found that for periodic boundary conditions, only the $1$- and $(2n+1)$-dimensional irreducible representations are present in the ground space, and thus the ground-space degeneracy grows linearly in the system size.

Our results generalize to other simple Lie algebras $\mathfrak{g}$. We have obtained an explicit formula for the ground-space degeneracy in the open boundary case; it is a polynomial of degree $d-1$ in the system size where $d=\dim\mathfrak{g}$.
To summarize, we provided a broad family of examples of injective MPS with parent Hamiltonians that have short interaction range and a large ground-space degeneracy, complementing the results of \cite{schuch2025} showing that generically these types of Hamiltonians have unique ground states. 

We expect that the ground-space degeneracy of the generalized model with periodic boundary conditions can be determined by similar methods. It would also be interesting to generalize the construction to two dimensions and analyze the ground-space structure and spectral properties of the resulting models.

\section*{Acknowledgements}

We would like to thank Eugene Dumitrescu and Michael Ragone for insightful discussions.
BNB was supported by the U.S. Department of Energy, Advanced Scientific
Computing Research, under contract number DE-SC0025384.  NS and AM
acknowledge funding by the Austrian Science Fund FWF (Grant
Nos.~\href{https://doi.org/10.55776/COE1}{10.55776/COE1},
\href{https://doi.org/10.55776/F71}{10.55776/F71},
\href{https://doi.org/10.55776/P36305}{10.55776/P36305}), by the European Union
-- NextGenerationEU, and by the European Union’s Horizon 2020 research and
innovation programme through Grant No.~863476 (ERC-CoG \mbox{SEQUAM}).

\printbibliography

\end{document}